\documentclass[sigconf,twocolumn]{acmart}

\AtBeginDocument{%
  \providecommand\BibTeX{{%
    \normalfont B\kern-0.5em{\scshape i\kern-0.25em b}\kern-0.8em\TeX}}}

\usepackage{csquotes}
\usepackage{amsmath,amsfonts}
\usepackage{array}     
\usepackage{comment}  
\usepackage{xcolor}
\usepackage{hyperref}    
\usepackage{float}
\usepackage{url}
\usepackage{tikz}
\usepackage{soul}
\usepackage{siunitx}
\usepackage[utf8]{inputenc}
\usepackage{caption}

\floatstyle{ruled}
\newfloat{algo}{tp}{lop}
\floatname{algo}{Algorithm}
\usepackage[noend]{algpseudocode}
\algrenewcommand\textproc{\sffamily}
\algrenewcommand\algorithmicindent{1em}
\usepackage{flushend}  
\usepackage{mathtools} 
\usepackage{booktabs}  
\usepackage{tikz-cd}
\usepackage{nicefrac}
\usepackage[linesnumbered,lined,boxed]{algorithm2e}

\def\Nbb{\mathbb{N}}

\renewcommand{\leq}{\leqslant}
\renewcommand{\geq}{\geqslant}

\makeatletter
\def\mproof#1{%
    \trivlist
    \item[%
        \hskip 10\p@
        \hskip \labelsep
        {\sc #1.}%
    ]
    \ignorespaces
}
\def\mqed{%
    \unskip
    \kern 10\p@
    \hfill
    \begingroup
        \unitlength\p@
        \linethickness{.4\p@}%
        \framebox(6,6){}%
    \endgroup
    \global\@qededtrue
}
  \newtheorem{theorem}{Theorem}[section]
  \newtheorem{proposition}[theorem]{Proposition}
  \newtheorem{lemma}[theorem]{Lemma}
  \newtheorem{remark}[theorem]{Remark}
  \theoremstyle{definition}

\makeatother
{{\hfill{$\square$}\bgroup\parfillskip=0pt\par\egroup}}

\newenvironment{boldthm}{\bfseries\begin{theorem}}{\end{theorem}\mdseries}

\newcommand{\quotient}[2]{%
\leavevmode%
\kern-.1em
\raise.2ex\hbox{\(#1\)}\kern-.1em%
/%
\kern-.1em\lower.25ex\hbox{\(#2\)}}

\usepackage[textwidth=1.7cm]{todonotes}

\begin{document}
\fancyhead{}


\title[Computing Characteristic Polynomials of $p$-Curvatures in Average
Polynomial Time]{Computing Characteristic Polynomials of $p$-Curvatures \\
in Average Polynomial Time}


\author{Rapha{\"e}l Pag{\`e}s}
\affiliation{%
  \institution{IMB, Universit{\'e} de Bordeaux, France}
  \country{}
}
\email{raphael.pages@u-bordeaux.fr}

\renewcommand{\shortauthors}{Rapha{\"e}l Pag{\`e}s}

\begin{abstract}
We design a fast algorithm that computes, for a given linear differential operator with coefficients in $\mathbb{Z}[x]$, all 
the characteristic polynomials of its $p$-curvatures, for all primes~$p< N$, 
in asymptotically quasi-linear bit complexity in $N$.
We discuss implementations and applications of our algorithm.
We shall see in particular that the good performances of our algorithm
are quickly visible.
\end{abstract}



\maketitle

\section{Introduction}\label{sec:intro}

The study of differential equations is a large part of mathematics which finds
applications in many fields, particularly in physical sciences. Although the
classical study of differential equations concerns essentially functions of
real or complex variables, those equations can also be studied in an
algebraic way.
The functions in calculus get replaced by the elements of a so-called
\emph{differential ring}, and the ``set of differential equations'' is endowed
with a ring structure. The resulting formalism is more flexible than that of
calculus and makes it possible to study problems in positive
characteristic.

In the algebraic context, the most relevant questions about a
linear differential system $Y'=AY$, with~$A$ a matrix with coefficients in
$\mathbb{Q}(x)$, differ a little from those in calculus. For example we
may ask ourselves if such a system has an algebraic basis of solutions.
This problem is especially difficult, though
decidable, as
was shown by Singer in \cite{MFSinger79} (see also~\cite{BCDW16}).

However, such a system can be reduced modulo~$p$ for any prime~$p$ not
dividing the denominators of the matrix. Thus we can consider reductions
modulo~$p$ of a given linear differential system. This construction turns out
to be useful.
Indeed if a system has an algebraic basis of solutions in
characteristic~$0$, then its reduction modulo~$p$ also has one for almost
all primes~$p$. The well-known Grothendieck-Katz conjecture~\cite{Katz82} 
states that this is in fact an equivalence.

Thus it is very interesting for a
given linear differential system in characteristic~$0$ to be able to
determine if their
reduction modulo~$p$ have a basis of algebraic solutions (or, more
generally, to determine the dimension of its space of algebraic solutions)
for a large amount of
primes~$p$, even if this only
has heuristic applications for the time being. However, effective versions
of the Grothedieck-Katz conjecture would turn this
heuristic into a complete algorithm.\par
The resolution of this problem in positive characteristic is much easier
than in characteristic~$0$ thanks to an invariant of linear differential
systems in characteristic~$p$: the $p$-curvature. This invariant is a
linear map, whose kernel has the same dimension as the space of
algebraic solutions of $Y'=AY$. Moreover, it is ``easily computable'', as
its matrix is the $p$-th matrix $A_p$ of the recursive sequence
\begin{equation}\label{eqn:reccursive_sequence}A_1=-A \quad \text{ and } \quad A_{i+1}=A_{i}'-A\cdot A_i \;
\text{ for } \; i \geq 1.
\end{equation}
In this paper we are interested in computing the
characteristic polynomials of the $p$-curvatures of a linear differential
operator with coefficients in $\mathbb{Z}[x]$ for a whole range of primes
$p< N$.
This information contains
an upper bound on the dimension of the kernel of the $p$-curvatures. It also
enables us to tell whether the
$p$-curvatures are nilpotent. This is interesting since
Chudnovsky's theorem, of which a formulation can be found in
\cite[Section~VIII.1, Theorem~1.5]{DwGeSu94}, states that the minimal
operator making a $G$-function vanish is globally nilpotent. As being
globally nilpotent is quite an uncommon property, this provides a robust heuristic test when trying to post-certify a guessed annihilating differential operator.\par
The naive approach to this problem consists in computing the $p$-curvature
with the recursive sequence~\eqref{eqn:reccursive_sequence} and then computing its characteristic polynomial.
This strategy is sometimes referred to as \emph{Katz's algorithm}~\cite[p.
324]{PuSi03} and outputs the result in~$\tilde{O}(p^2)$ bit operations (in this paper the notation~$\tilde{O}$
will have the same meaning as $O$ except we neglect logarithmic factors). Bostan, Caruso and Schost~\cite{BoCaSc14} brought back the computation of
the characteristic polynomial of the $p$-curvature to that of a factorial
of matrices, and presented an algorithm finishing in $\tilde{O}(\sqrt{p})$
bit operations. It is unknown if the $1/2$ exponent is optimal for this
problem. Indeed, the characteristic polynomial of the $p$-curvature is a
polynomial $P$ of degree $O(1)$ in $x^p$, and it is still unknown whether
$P$ is computable in polynomial time in $\log(p)$.

In this paper, we build upon~\cite{BoCaSc14} to design an algorithm
computing, for a given differential operator, almost
all of the characteristic polynomials of its $p$-curvatures, for all
primes~$p<N$, in quasi-linear, thus quasi-optimal, time in~$N$. This is a
significant improvement over previous algorithms for the given task, since
the iterations of \emph{Katz's algorithm} and of the algorithm
from~\cite{BoCaSc14} only terminate in respectively $\tilde{O}(N^3)$ and
$\tilde{O}(N^{3/2})$ bit operations.\par
Since the number of
primes smaller than~$N$ is also quasi-linear in~$N$, this means that the
average time spent on the computation of one characteristic polynomial is
polynomial in~$\log(N)$. It is important to note that  ``average'' here is
meant as average over the range of primes, and definitely not over the set
of operators (even of fixed degree and order).

To achieve this goal, we reuse an idea of
Costa, Gerbicz and Harvey, who designed
an algorithm computing $(p{-}1)!\bmod p^2$
for all primes~$p$ less than~$N$ in quasi-linear time in~$N$~\cite{CoGeHa14}.
This algorithm was 
originally designed to search for the so-called \emph{Wilson primes}, but 
it soon found many
applications, for instance in counting points on
curves~\cite{Harvey14}.

We begin this article with a quick reminder of the theoretical facts about
differential operators which make our algorithm possible. We then
present our algorithm and evaluate its complexity to see that it is indeed
quasi-linear in~$N$. Lastly we present the results of our implementation of
the algorithm in the computer algebra software \emph{SageMath}.\\
       
{\bf Acknowledgements.} 
This work was supported by
\textcolor{magenta}{\href{https://specfun.inria.fr/chyzak/DeRerumNatura/}{DeRerumNatura}}
ANR-19-CE40-0018 and CLap--CLap ANR-18-CE40-0026-01.
I address special thanks to my PhD thesis advisors,
Alin Bostan and Xavier Caruso who helped me during the preparation of this article, whose roots are in my Master’s thesis~\cite{master}. I also warmly thank the reviewers for their relevant and numerous
comments and the amazing amount of work they put on this paper.


\section{Differential operators}\label{sec:theory}

In this section, we outline the theoretical aspects necessary to our algorithm
by following the exposition of \cite{BoCaSc14} (to which we refer for more
detailed explanations) and extending the results of \emph{loc. cit.} to
characteristic~$0$. All results in Sections~\ref{euler_op}
and~\ref{p_curv_def} come from~\cite{BoCaSc14}. Besides, proofs were added
when they were not given in \emph{loc. cit.}

Let~$\mathcal{R}$ be either $R[x]$ or $R(x)$, with $R=\mathbb{Z}$ or
$\mathbb{F}_p$, equipped with their usual derivation $f\mapsto f'$. 
Throughout this article we will
study the ring of differential operators with coefficients in~$\mathcal{R}$,
which we denote by~$\mathcal{R}\langle\partial\rangle$.
The elements of $\mathcal{R}\langle\partial\rangle$ are polynomials in
$\partial$ of the form
\[f_n\partial^n+f_{n-1}\partial^{n-1}+\cdots+f_1\partial+f_0\] with
$f_i\in\mathcal{R}$. The (noncommutative) multiplication in this
ring is deduced from the Leibniz rule
$\partial f=f\partial+f'$ for all elements~$f$ of~$\mathcal{R}$.

\subsection{Euler and integration operators}\label{euler_op}
In Sections~\ref{euler_op} and~\ref{p_curv_def} we will only consider the case
$R=\mathbb{F}_p$. We study the Euler operator $x\partial$. One can show
that 
\[
  \partial\cdot(x\partial)=(x\partial+1)\cdot\partial
  \quad\text{and}\quad
  x\cdot(x\partial)=(x\partial-1)\cdot x.\]\par
  We introduce a new variable~$\theta$ and consider the noncommutative
  ring $\mathbb{F}_p[\theta]\langle\partial\rangle$ (resp.
  $\mathbb{F}_p(\theta)\langle\partial\rangle$) 
  whose elements are polynomials in the variable~$\partial$ with coefficients
  in $\mathbb{F}_p[\theta]$ (resp. $\mathbb{F}_p(\theta)$), with 
  multiplication deduced from the rule
  $\partial\theta=(\theta+1)\partial$.

  We now want to rewrite operators in the variable~$x$ as operators in the
  variable~$\theta$ with the association $\theta\mapsto x\partial$.
  In order to do this, we introduce the integration operator~$\partial^{-1}$
  and the algebras $\mathbb{F}_p[x]\langle\partial^{\pm 1}\rangle$ (resp.
  $\mathbb{F}_p(x)\langle\partial^{\pm 1}\rangle$) of Laurent
  polynomials in the variable~$\partial$ with coefficients in $\mathbb{F}_p[x]$ (resp.
  $\mathbb{F}_p(x)$).
  The same can be done in the variable~$\theta$.
  \begin{proposition}[{\cite[Section~2]{BoCaSc14}}]
    The rings $\mathbb{F}_p[x]\langle\partial^{\pm
    1}\rangle\subset\mathbb{F}_p(x)\langle\partial^{\pm 1}\rangle$ (resp.
     $\mathbb{F}_p[\theta]\langle\partial^{\pm1}\rangle 
    \subset\mathbb{F}_p(\theta)\langle\partial^{\pm 1}\rangle$) of
    Laurent polynomials in the variable $\partial$ are all well defined.
  Furthermore, the multiplication satisfies $\partial^{-1}f=\sum_{i=0}^{p
    -1}(-1)^if^{(i)}\partial^{-i-1}$ for all~$f\in \mathbb{F}_p(x)$, 
    and $\partial^ig(\theta)=g(\theta+i)\partial^i$ for all $ g\in
    \mathbb{F}_p(\theta)
    $ and $i\in\mathbb{Z}.$
  \end{proposition}
  \begin{proof}
    One can show that~$\partial^p$ is central in
    $\mathbb{F}_p(\theta)\langle\partial\rangle$. This is also the case in
    $\mathbb{F}_p(x)\langle\partial\rangle$ since $f^{(p)}=0$ for all
    $f\in\mathbb{F}_p(x)$, and thus~$\partial^p f=\sum_{i=0}^p
  \binom{p}{i}f^{(i)}\partial^{p-i}=f\partial^p
      +f^{(p)}=f\partial^p$.\par
    It follows that we only need to invert
    the central element~$\partial^p$ of both sets of rings, which can be
    done the same way as commutative localization.\par

    The first relation comes from the fact that
    $\partial^{-1}f=\partial^{p-1}f\partial^{-
  p}$ and~$\binom{p-1}{i}\equiv (-1)^i\mod p$ and the
      second one is trivial.
  \end{proof}
  \begin{theorem}[{\cite[Section~2.2]{BoCaSc14}}]
  \label{phi_p}
    The following induces an isomorphism of $\,\mathbb{F}_p$-algebras:
  \[\begin{array}{rcccl}
    &\mathbb{F}_p[x]\langle\partial^{\pm
    1}\rangle&\overset{\sim}{\leftrightarrow}&
    \mathbb{F}_p[\theta]\langle\partial^{\pm 1}\rangle&\\
    \varphi_p:&x&\mapsto& \theta\partial^{-1}&\\
    &x\partial&\mapsfrom&\theta&:\psi_p\\
    &\partial&\leftrightarrow&\partial&
  \end{array}\]
  \end{theorem}
  \begin{proof}
    It is enough to check that 
    $\varphi_p(\partial)\varphi_p(x)=\varphi_p(x)\varphi_p(\partial)+1$ and
    $\psi_p(\partial)\psi_p(\theta) =(\psi_p(\theta)+1)\psi_p(\partial)$
    to see that $\varphi_p$ and $\psi_p$ are well defined. 
    We check that $\psi_p$ and $\varphi_p$ are invertible by checking
    that~$\psi_p\circ \varphi_p$ (resp. $\varphi_p\circ \psi_p$)  is the only morphism
    mapping~$x$ to~$x$ (resp.~$\theta$ to~$\theta$)
    and $\partial$ to $\partial$.
  \end{proof}
  \begin{remark}[{\cite[Section~2.2]{BoCaSc14}}]
    The element $(x{+}1)\partial$ is invertible in
    $\mathbb{F}_p(x)\langle\partial^{\pm
    1}\rangle$ but $\varphi_p((x+1)\partial)=\theta+\partial$ is not
    invertible in~$\mathbb{F}_p(\theta)\langle\partial^{\pm 1}\rangle$. As
    such, $\varphi_p$ does not extend to an isomorphism
    \vspace{-3pt}
    $$\mathbb{F}_p(x)\langle\partial^{\pm
    1}\rangle \rightarrow \mathbb{F}_p(\theta)\langle\partial^{\pm
    1}\rangle.$$
  \end{remark}
  One can show that $\mathbb{F}_p[\theta^p-\theta]\langle\partial^{\pm
  p}\rangle$ is the center of
  $\mathbb{F}_p[\theta]\langle\partial^{\pm 1}\rangle$ and that 
  $\varphi_p^{-1}(\theta^p-\theta)=x^p\partial^p$. This will be useful
    later on.

\subsection{Operators and $p$-curvature}\label{p_curv_def}

We recall that for $L\in \mathbb{F}_p(x)\langle\partial\rangle$, the 
left multiplication by the operator
  $\partial^p$ defines an $\mathbb{F}_p(x)$-linear endomorphism of
    $\nicefrac{\mathbb{F}_p(x)\langle\partial\rangle}{\mathbb{F}_p(x)
    \langle\partial\rangle L}$
  since $\partial^p$ is a central element. We define the $p$-curvature of $L$ as being this $\mathbb{F}_p(x)$-linear endomorphism
  or, for computational purposes, its matrix in the canonical basis
  $(1,\partial,\partial^2,\ldots)$, which we
  denote by~$A_p(L)$.
\begin{remark}
  It follows from the definition that the $p$-curvature of a differential
  operator~$L$ does not change if~$L$ is multiplied on the left by an 
  element of~$\mathbb{F}_p(x)$. Though Algorithm~\ref{finalalgorithm}
  presented in  Section~\ref{mainalgorithm} will
  work for operators in $\mathbb{Z}[x]\langle\partial\rangle$ for
  convenience,
  this remark allows us to say that it in fact works for all operators in
  $\mathbb{Q}(x)\langle\partial\rangle$.
\end{remark}

  \label{ssec:defcharpoly}
  As we did for operators with coefficients in $\mathbb{F}_p(x)$, we define the
  $p$-curvature of an operator~$L$ with coefficients in
  $\mathbb{F}_p(\theta)$ as
  the~$\mathbb{F}_p(\theta)$-linear endomorphism of
  $\nicefrac{\mathbb{F}_p(\theta)\langle\partial\rangle}{\mathbb{F}_p
  (\theta)\langle\partial\rangle\cdot
  L}$ induced by the left multiplication by~$\partial^p$, and we denote 
  by~$B_p(L)$ its matrix in
  the canonical basis $(1,\partial,\partial^2,\ldots)$. By
  \cite[Lemma~2.3]{BoCaSc14} which is proved by a straightforward
  computation, if $B(L)(\theta)$ is the companion
  matrix of~$L$ then
  \[B_p(L)=B(L)(\theta)\cdot B(L)(\theta+1)\cdots B(L)(\theta+p-1).\]
  As we are interested in computing the characteristic polynomial of
  the
  $p$-curvature we introduce the following (\emph{cf}~\cite[Section~3]{BoCaSc14}):\\
  Let $L_x\in \mathbb{F}_p(x)\langle \partial\rangle$, and $L_\theta\in
  \mathbb{F}_p(\theta)\langle\partial\rangle$. We
  denote their respective leading coefficients by $l_{x}\in\mathbb{F}_p(x)$
  and $l_{\theta}\in\mathbb{F}_p(\theta)$ respectively and define two new operators:
  \begin{align*}
    \Xi_{x,\partial}(L_x)&:=l_{x}^p\chi(A_p(L_x))(\partial^p)\\
    \Xi_{\theta,\partial}(L_\theta)&:=\left(\prod_{i=0}^{p-1}l_{\theta}
    (\theta+i)\right)\chi(B_p(L_\theta))(\partial^p)
  \end{align*}
  where $\chi(M)$, for a square matrix~$M$, is its characteristic
  polynomial.
  \begin{remark}
    Depending on the context, we may write~$\Xi_{x,\partial,p}$
    and~$\Xi_{\theta,\partial,p}$ if we want to specify the characteristic.
  \end{remark}
  \begin{proposition}[{\cite[Section~3.1]{BoCaSc14}}]
    The maps $\Xi_{x,\partial}$ and $\Xi_{\theta,\partial}$ are 
    multiplicative and can thus be extended to
    maps on $\mathbb{F}_p(x)\langle\partial^{\pm 1}\rangle$ and
    $\mathbb{F}_p(\theta)\langle\partial^{\pm 1}\rangle$ respectively.
  \end{proposition}
  \begin{proof}
    Let
    $D:=\mathbb{F}_p(x)\langle\partial\rangle$ (resp.
    $D:=\mathbb{F}_p(\theta)\langle\partial\rangle$) and $L_1,L_2\in D$. The right 
    multiplication by~$L_2$
    induces a map $\zeta_1:\nicefrac{D}{DL_1}\rightarrow
    \nicefrac{D}{DL_1L_2}$. There is also a canonical map
    $\zeta_2:\nicefrac{D}{DL_1L_2}\rightarrow\nicefrac{D}{DL_2}$. We check that
    \[0\rightarrow
    \nicefrac{D}{DL_1}\xrightarrow{\zeta_1}\nicefrac{D}{DL_1L_2}
    \xrightarrow{\zeta_2}\nicefrac{D}{DL_2}\rightarrow 0\]
    is an exact sequence. Furthermore the left multiplication by
    $\partial^p$ induces an endomorphism of this exact sequence. It follows
    that in a suitable basis, the matrix of the $p$-curvature of~$L_1L_2$ is
    an upper triangular block matrix, with the upper left block being the
    matrix of the $p$-curvature of~$L_1$ and the bottom right block, that
    of~$L_2$. The multiplicativity immediately follows.
    We extend those applications by setting
    $\Xi_{x,\partial}(L\partial^{-n}) = \Xi_{x,\partial}(L)\Xi_{x,
    \partial}(\partial)^{-n}$ (resp. $\Xi_{\theta,\partial}$) for all $n$
    and all operators $L$.
  \end{proof}
  \begin{boldthm}[{\cite[Section~3]{BoCaSc14}}]\label{formofresult}$\;$  
    \begin{itemize}
      \item The map $\Xi_{x,\partial}$ (resp. $\Xi_{\theta,\partial}$)
        takes its values in $\mathbb{F}_p(x^p)[\partial^{\pm  p}]$ (resp.
        $\mathbb{F}_p(\theta^p-\theta)[\partial^{\pm p}]$).
      \item Those two maps send an operator with polynomial
        coefficients to an operator with polynomial coefficients.
      \item
    The following diagram commutes:
    \[\begin{tikzcd}
  \mathbb{F}_p[x]\langle\partial^{\pm 1}\rangle\arrow[d,"\Xi_{x,\partial}"]
  \arrow[r,"\substack{\varphi_p\\\sim}"]&\mathbb{F}_p[\theta]\langle\partial^{\pm 1}\rangle
  \arrow[d,"\Xi_{\theta,\partial}"]\\
  \mathbb{F}_p[x^p][\partial^{\pm  p}]\arrow[r,"\substack{\varphi_p\\\sim}"]&
  \mathbb{F}_p[\theta^p-\theta][\partial^{\pm p}]
  \end{tikzcd}\]
    \end{itemize}
  \end{boldthm}
    This is the main result that makes our algorithm possible. 
 Theorem~\ref{formofresult} is interesting since it brings back the computation of the
 characteristic polynomial of the $p$-curvature to that of the ``factorial of
 matrices'' $B_p(L)$, and can thus be computed using factorial computation
 methods.
 \subsection{Extension to integral coefficients}
Although the $p$-curvature is defined for operators of
$\mathbb{F}_p(x)\langle\partial\rangle$, we can define the $p$-curvature of
an element of
$\mathbb{Z}[x]\langle\partial\rangle$, since the canonical morphism
$\mathbb{Z}\rightarrow\mathbb{F}_p$ induces a ring homomorphism
\[\mathbb{Z}[x]\langle\partial\rangle\rightarrow\mathbb{F}_p[x]\langle\partial\rangle.\]
 Our goal is to compute, for a differential operator with coefficients
 in~$\mathbb{Z}[x]$, the characteristic polynomials of its $p$-curvatures,
 for nearly all primes~$p$ up to a certain integer~$N$, in~$\tilde{O}(N)$ bit
 operations. 

\begin{proposition}
  The rings $\mathbb{Z}[x]\langle\partial^{\pm 1}\rangle$ and
  $\mathbb{Z}[\theta]\langle\partial^{\pm 1}\rangle$ (analogous to those of
  Section~\ref{euler_op}) are well defined and we
  have an
  isomorphism $\varphi:\mathbb{Z}[x]\langle\partial^{\pm
  1}\rangle\xrightarrow{\sim}\mathbb{Z}[\theta]\langle\partial^{\pm
  1}\rangle$ defined in a similar manner to $\varphi_p$ (see Theorem~\ref{phi_p}).
\end{proposition}
\begin{proof}
  It is enough to check that  the multiplicative part 
  $S=\{\partial^n|n\in\mathbb{N}\}$ is a right denominator set of
  the ring $\mathbb{Z}[\theta]\langle\partial\rangle$ (see 
  \cite[Section~10A]{Lam99}). Since this ring has no nontrivial zero divisor, we
  only have to check that $S$ is right permutable, that is to say that
  \[\forall g\in\mathbb{Z}[\theta]\langle\partial\rangle,\forall
  n\in\mathbb{N},\exists g_1\in\mathbb{Z}[\theta]\langle\partial\rangle,\exists
  n_1\in\mathbb{N},g\partial^{n_1}=\partial^{n}g_1.\]
  This is the case since for all $n\in\mathbb{N}$ and all
  $g\in\mathbb{Z}[\theta]$, $\partial^n g(\theta-n)=g\partial^n$
  and the fact that $\mathbb{Z}[\theta]\langle\partial^{\pm
  1}\rangle$ is well defined follows by additivity.\par
  The same can be done for operators with coefficients in the variable~$x$.
  Let $f\in \mathbb{Z}[x]$ and suppose that $f^{(n_1)}=0$. Then
  \[f\partial^{n_1+1}=\partial\sum_{k=0}^{n_1-1}(-1)^kf^{(k)}
  \partial^{n_1-k}.\]
  Now by induction on~$i$, we show that for all $n_1\in\Nbb$, all
  $i\in\Nbb^*$
  and all $f\in\mathbb{Z}[x]$ such that $f^{(n_1)}=0$, there
  exists~$f_i\in\mathbb{Z}[x]\langle\partial\rangle$ such that
  $f\partial^{n_1+i}=\partial^if_i$. We then conclude by additivity,  
  which yields the fact that $\mathbb{Z}[x]\langle\partial^{\pm
  1}\rangle$ is well defined. We show that~$\varphi$ is an isomorphism the same way we did
for~$\varphi_p$.\end{proof}
By denoting $\pi_{p}:\mathbb{Z}\rightarrow\mathbb{F}_p$ the canonical
reduction modulo $p$, we can easily see that $\pi_{p}\circ
\varphi=\varphi_p\circ \pi_{p}$ (where we extend naturally~$\pi_p$ to
suitable rings of operators).
This enables us, for a
given operator in $\mathbb{Z}[x]\langle\partial\rangle$, to compute the
characteristic polynomials of its $p$-curvatures, by computing the
isomorphism~$\varphi$ before the reduction modulo~$p$. We will now see how
to use this fact.

\section{Main algorithm}
In this section, we present our algorithm and estimate its complexity. We
denote by $2\leq\omega\leq 3$ an exponent of matrix multiplication. From
\cite{AlVW21}, we know that we can take $\omega<2.3728596$.
We will also have to address the cost of computing characteristic
polynomials. Let
us denote $\Omega_1\in\mathbb{R}^*_+$ such that the computation of the
characteristic polynomial of a square matrix of size~$m$ with coefficients in
a ring $R$ can be done in $\tilde{O}(m^{\Omega_1})$ arithmetic operations in
$R$.
From \cite[Section~6]{KaVi04}, we know that it is theoretically possible to
take $\Omega_1\simeq 2.697263$.
 Finally, throughout this section, we assume that any
two polynomials of degree~$d$ over a ring~$R$
(resp. integers of bit size $n$) can be multiplied in $\tilde{O}(d)$
operations in $R$ (resp. $\tilde{O}(n)$ bit operations); FFT-like algorithms
allow for these complexities~\cite{CaKa91,HaHo21}.
We now give an outline of our algorithm.
\smallskip
\begin{algo}
  \begin{flushleft}
    \emph{Input:} $L_x\in \mathbb{Z}[x]\langle\partial\rangle$, $N\in\mathbb{N}$\\
    \emph{Output:} A list of the characteristic polynomials of the
    $p$-curvatures of~$L_x$, for all primes $p$ with $p< N$ except a
    finite number not depending on~$N$.
  \end{flushleft}
  \BlankLine
    \begin{enumerate}
      \item Name $l_x$ the leading coefficient of $L_x$.
      \item Compute
        $L_\theta:=\varphi(L_x)\in\mathbb{Z}[\theta]\langle\partial^{\pm 1}\rangle$.
\item Name $l_{\theta}$ the leading coefficient of $L_\theta$.
\item Compute $\mathcal{P}_{l_\theta}$, the list of all primes $p< N$
  which do not divide $l_\theta$.
\item  Construct $B(L_\theta)$.
\item  Compute
  $\left(\prod_{i=0}^{p-1}l_\theta(\theta+i)\right) \bmod p$ for all $p\in
        \mathcal{P}_{l_\theta}$.
\item  Compute $B(L)(\theta)\cdots B(L)(\theta+p-1) \bmod p$ for all $p\in
  \mathcal{P}_{l_\theta}$.
\item  Deduce all the $\Xi_{\theta,\partial,p}(L_\theta)$, for $p \in
  \mathcal{P}_{l_\theta}$.
\item  Deduce all
  $\chi(A_p(L_x))=l_x^{-p}\varphi_p^{-1}(\Xi_{\theta,\partial,p}(L_\theta))$,
        for $p \in \mathcal{P}_{l_\theta}$.
    \end{enumerate}
\end{algo}
\begin{remark}
  We only do the computation for the primes which do not divide the leading
  coefficient of $L_\theta$ because for those which do, the companion matrix
  of
  its reduction modulo~$p$ is not the reduction modulo~$p$ of its companion
  matrix.
\end{remark}

\begin{lemma}\label{degreeofresult}
  Let $L_\theta\in\mathbb{F}_p[\theta]\langle\partial\rangle$ be an operator
  with coefficients of degree at most $d\in\Nbb$. Then
  $\Xi_{\theta,\partial}(L_\theta)$ has coefficients of degree at most~$dp$.
\end{lemma}
\begin{proof}
  See \cite[Lemma~3.9]{BoCaSc14}.
\end{proof}
From Lemma~\ref{degreeofresult}, we deduce that at the end of step~(8) we
have a list of (lists of) polynomials of degree linear in~$p$, which means
that the bit size of
the output of this
step is quadratic in~$N$. This seems to remove all hope of ending up
with a quasi-linear algorithm. Fortunately those polynomials lie
in $\mathbb{F}_p[\theta^p-\theta]$( see Theorem~\ref{formofresult}). Thus each of them can be
represented by data of bit size $O(d\log(p))$. We explain how in
Section~\ref{reverse_iso_section}.
\begin{remark}\label{arragementresult}
  This problem is also present at the end of step~(9), but is easy to solve
  as we only need to determine the coefficients of~$x^i$ when $i$ is a multiple
  of~$p$. Thus we in fact compute polynomials $P_p\in\mathbb{F}_p[x,Y]$
  such that $P_p(x^p,Y)=\chi(A_p(L))$ for all $p< N$.
\end{remark}

\subsection{Reverse isomorphism, computation modulo
$\theta^{d+1}$}\label{reverse_iso_section}
We know from Theorem~\ref{formofresult} that for
$L_\theta\in\mathbb{F}_p[\theta]\langle\partial\rangle$, 
the operator~$\Xi_{\theta,\partial}(L_\theta)$ has coefficients in 
$\mathbb{F}_p[\theta^p-\theta]$.
\begin{lemma}\label{inverse_iso_comput}
  Let $Q\in \mathbb{F}_p[\theta^p-\theta]$ be a polynomial of degree~$d$
  in~$\theta^p-\theta$ with $d< p$. Write:
  \[
    Q=\sum_{i=0}^d q_i(\theta^p-\theta)^i\quad\text{and}\quad
    Q=\sum_{i=0}^{dp} q'_i \theta^i.
  \]
  For all $i\leq d$, we have $q_i=(-1)^i q'_i$.
\end{lemma}
\begin{proof} This comes from the fact that $(-1)^i\theta^i$ is the only
  monomial of degree less than~$p$ in $(\theta^p-\theta)^i$.
\end{proof}
When~$p$ is strictly greater than~$d$, it
follows that we only need to compute the $\Xi_{\theta,\partial,p}$ modulo
$\theta^{d+1}$ where~$d$ is the highest degree of the coefficients of the
operator (in both variables~$x$ or~$\theta$), as one can see in
Algorithm~\ref{reverse_iso_algo}. We deduce the following lemma whose proof
is obvious.
\begin{algo}
  \begin{flushleft}
    \emph{Input:} $Q_\theta\in\mathbb{F}_p[\theta^p-\theta][Y]$, of
    degree $m$ in $Y$ and degree at most $dp$ in $\theta$,
    known
    modulo $\theta^{d+1}$.\\
    \emph{Output:} $Q_x\in\mathbb{F}_p[x,Y]$ such that
    $Q_x(x^p,\partial^p)=\varphi_p^{-1}(Q_\theta(\partial^p))$.
  \end{flushleft}
  \BlankLine
  \begin{enumerate}
    \item $Q_x\leftarrow 0$.
    \item For all $i\leq m$:
      \begin{enumerate}
        \item Let $Q_{\theta,i}$ be
          the coefficient of $\partial^i$ of $Q_\theta$ and write
          $Q_{\theta,i}=\sum_{j=0}^{d}q_{i,j} \theta^j+O(\theta^{d+1})$.
        \item $Q_x\leftarrow Q_x + \sum_{j=0}^{d}(-1)^j q_{i,j}
          x^jY^{i+j}$.
      \end{enumerate}
    \item \emph{Return:} $Q_x$.
    \end{enumerate}
    \caption{reverse\_iso}
    \label{reverse_iso_algo}
\end{algo}
\begin{lemma}
  If $Q_\theta\in\mathbb{F}_p[\theta^p-\theta][Y]$ is 
  of degree~$m$ in~$Y$ and~$dp$ in~$\theta$  with $d<p$, then
  Algorithm~\ref{reverse_iso_algo} computes $Q_x\in\mathbb{F}_p[x,Y]$ such
  that~$Q_x(x^p,\partial^p)=\varphi_p^{-1}(Q_\theta(\partial^p))$ in
  $O(dm\log(p))$ bit operations. 
\end{lemma}
\begin{remark}\label{remark_on_d}
  In fact we can still compute $\varphi_p^{-1}$ if $p\leq d$ while
  only knowing the operator modulo $\theta^{d+1}$ but
  this is more tedious since there is no nice formula. In that case, with
  notation as in
  Lemma~\ref{inverse_iso_comput}, we have
  $q_i'=\sum_{k=0}^{\lfloor i/(p-1)\rfloor}(-1)^{i-kp}\binom{i-k(p-
  1)}{k}q_{i-k(p-1)}.$\\
  This relation is easily invertible since it is given by a triangular matrix
  with no zero on the diagonal.
\end{remark}

\subsection{Translation before the computation}

From the results of the previous subsection, 
we know that we only need to determine $\Xi_{\theta,\partial}$ modulo a
small power of $\theta$. Unfortunately, the companion matrix of an operator
in $\mathbb{F}_p[\theta]\langle\partial\rangle$,
even if the operator has polynomial coefficients, usually has its
coefficient in~$\mathbb{F}_p(\theta)$. In \cite{BoCaSc14}, the authors
solve this issue by injecting~$\mathbb{F}_p(\theta)$ in
$\mathbb{F}_p((\theta))$ and computing modulo a slightly higher power of $\theta$. In order
to minimize the degree of the polynomials used in the computation, we
take a different approach based on the following proposition.

\begin{proposition}
\label{prop:translate}
  Let $a\in\mathbb{F}_p$. We denote by
  $\tau_a:\mathbb{F}_p[x]\rightarrow \mathbb{F}_p[x]$ the shift
  automorphism $Q\mapsto Q(x+a)$.
  This automorphism extends to automorphisms of 
  $\mathbb{F}_p[x]\langle\partial\rangle$ and $\mathbb{F}_p[x,Y]$. Then
  $$\tau_a\circ \chi(A_p)=\chi(A_p)\circ
  \tau_a.$$
\end{proposition}
\begin{proof}
  We know that $\tau_a(f)'=\tau_a(f')$ for all $f\in\mathbb{F}_p[x]$. We can
  thus  extend $\tau_a$ to $\mathbb{F}_p[x]\langle\partial\rangle$. Now,
  since for any $L$, the operator~$\tau_a(L)$ has the same order as $L$, 
  we get that
  $A(\tau_a(L))=\tau_a(A(L))$ (where~$A(L)$ is the companion matrix of $L$).
  Now with the relation between $\tau_a$ and
  derivation we recursively extend that equality using 
  (\ref{eqn:reccursive_sequence}) to get
  $\tau_a(A_p(L))=A_p(\tau_a(L))$. Since $\tau_a$ is an endomorphism, the
  result follows.
\end{proof}
From Proposition~\ref{prop:translate}, we deduce that we can shift an
operator before computing the
characteristic polynomials of its $p$-curvatures, and do the opposite
translation on those to get the desired result. It is especially useful
because of the following lemma.
\begin{lemma}\label{leadingcoeff}
  Let $L_x\in\mathbb{Z}[x]\langle\partial\rangle$ be an operator and denote
  by $l_x\in\mathbb{Z}[x]$ its
  leading coefficient. If $l_x(0)\neq 0$ then $\varphi(L_x)$ has
  $l_x(0)\in\mathbb{Z}$ as its leading coefficient.
\end{lemma}
\begin{proof}
  A straightforward computation shows that
  $\varphi(x^i\partial^j)=p_i(\theta)\partial^{j-i}$ with $p_i(\theta)$
  being a
  polynomial only dependent on $i$ (and not on $j$). Thus the leading coefficient of
  $\varphi(L_x)$ can only come from the constant coefficient of $l_x$ if
  this one is not $0$.
\end{proof}
In our setting, the fact that $\varphi(L_x)$ has a constant leading coefficient 
means that its companion matrix (see \S\ref{ssec:defcharpoly}) has its coefficients in
$\mathbb{Q}[\theta]$, implying that we can do all the computations modulo $\theta^{d+1}$.
Lemma~\ref{leadingcoeff} shows that we can shift our starting operator by $a\in\mathbb{Z}$ where $a$ is not a
root of its leading coefficient to place ourselves in that setting.\par
Since translating back all the characteristic polynomials
(the $P_p$ in fact, see Remark~\ref{arragementresult})
at the end of the 
  computation is basically the same as translating a list of $O(Nm)
  $ univariate polynomials of degree $d$, it can be done in
  $\tilde{O}(Nmd)$ bit operations (for example
  with binary splitting), with $m$
  being the order of the operator and $d$ the maximum degree of its
  coefficients.

\subsection{Computing a matrix factorial modulo $p$ for a large amount of
primes $p$}
Let $M(\theta)\in \mathscr{M}_m(\mathbb{Z}[\theta])$ be a square matrix
of size $m$ with coefficients of
degree less than $d$. In this subsection we review the algorithm of
\cite{CoGeHa14,Harvey14} applied to the computation of the following matrix
factorial :
$$M(\theta) \cdot M(\theta+1)\cdots M(\theta+p-1)\mod (p,\theta^d)$$ for
all primes $p< N$.
Though very similar, the setting of \cite{Harvey14}
is slightly different from ours as it concerns only integer matrices
 and considers a different kind of products.
 For this reason, we
 prefer to take some time to restate the algorithm in full and, at the
 same time, take the opportunity to set up notations.\par
 Since the method of \cite{CoGeHa14} computes products of $p{-}1$ entries 
 modulo some power of $p$, we will compute $M(\theta+1)\cdots
 M(\theta+p-1)\mod (p,\theta^d)$ for all
$p$, and then left-multiply by $M(\theta)$.\par

Let $\eta:=\lceil\log_2(N)\rceil$. For all $i$ and $j$ with
$0\leq i\leq \eta$ and $0\leq j<2^i$, we denote
$U_{i,j}:=\left\{k\in\Nbb\,\left|\,j\frac{N}{2^i}<k\leq(j+
1)\frac{N}{2^i}\right.\right\}.$\par
It follows from the definition that for all $0\leq i<\eta$ and all $0\leq j<2^i$,
$U_{i,j}=U_{i+1,2j}\cup U_{i+1,2j+1}$. Furthermore, for $i=\eta$, the~$U_{i,j}$ are
either empty or a singleton.\\

From this, we introduce $T_{i,j}:=\prod_{k\in U_{i,j}}M(\theta+k)\mod
\theta^d$, with the
product being made by sorting elements of $U_{i,j}$ in ascending order, and
$S_{i,j}:=\prod_{\substack{p\in U_{i,j}\\p\text{ prime}}}p$. From now on, we
consider that the~$T_{i,j}$ are elements of
$\mathscr{M}_m(\nicefrac{\mathbb{Z}[\theta]}{\theta^d})$.
From the properties of $U_{i,j}$, we deduce that
$T_{i,j}=T_{i+1,2j}T_{i+1,2j+1}$ and $S_{i,j}=S_{i+1,2j}S_{i+1,2j+1}$.\\

These relations allow us to fill binary trees containing the~$T_{i,j}$
and~$S_{i,j}$ as their nodes from the bottom. Furthermore, filling those
trees is nothing more than computing a factorial by binary splitting, 
and keeping the intermediate steps in memory.\par

To see how to apply this to our problem we suppose that $p\in U_{\eta,j}$ for a
certain $j$. A direct computation gives:
\[\begin{array}{l}
  M(\theta+1)\cdot M(\theta+2)\cdots M(\theta+p-1) \mod (p,\theta^d) \smallskip \\
\hspace{7ex}= T_{\eta,0}T_{\eta,1}\cdots T_{\eta,j-1} \mod S_{\eta,j}.
\end{array}\]
This motivates the following definition: for all $i, j$ with $0\leq i\leq \eta$ 
and $0\leq j<2^i$, we set $W_{i,j}:=\prod_{k=0}^{j-1}T_{i,k}\mod S_{i,j}$.
The following lemma is easily checked.
\begin{lemma}
  For all $i$ and $j$ such that the following quantities are well defined, 
  $W_{i+1,2j}=W_{i,j}\mod S_{i+1,2j}$ and
  $W_{i+1,2j+1}=W_{i,j}T_{i+1,2j}\mod S_{i+1,2j+1}$.
\end{lemma}
Thus we can compute the $W_{\eta,j}$ by filling a binary tree from the top 
starting from $W_{0,0}=1$.
This proves the correctness of Algorithm~\ref{matrix_factorial}, while its
complexity is addressed in the next proposition.

\begin{algo}
  \begin{flushleft}
    \emph{Input:} $M(\theta)\in
    \mathscr{M}_m(\mathbb{Z}[\theta])$ with coefficients of degree less
    than $d$, $\mathcal{P}$ a list of primes smaller than $N$.\\
    \emph{Output:} A list containing $M(\theta)M(\theta+1)\cdots
    M(\theta+p-1)\mod (p,\theta^d)$ for all $p$ in $\mathcal{P}$.
  \end{flushleft}
  \BlankLine
  \begin{enumerate}
    \item $\eta\leftarrow \lceil \log_2(N)\rceil$.
    \item Fill $T_{\eta,\_}$ and $S_{\eta,\_}$.
    \item Compute the binary trees $T$ and $S$.
    \item $W_{0,0}\leftarrow 1$.
    \item For $i$ going from $0$ to $\eta-1$:
      \begin{enumerate}
        \item For $j$ going from $0$ to $2^i-1$:
          \begin{enumerate}
            \item $W_{i+1,2j}\leftarrow W_{i,j}\mod S_{i+1,2j}$.
            \item $W_{i+1,2j+1}\leftarrow W_{i,j}T_{i+1,2j}\mod
              S_{i+1,2j+1}$.
          \end{enumerate}
        \end{enumerate}
      \item Construct $\prod$ the list of $W_{\eta,j}$ 
        where $S_{\eta,j}\in \mathcal{P}$.
      \item Do the left multiplication by $M(\theta)$ on the elements of
        $\prod$.
      \item Return: $\prod$.
  \end{enumerate}
  \caption{matrix\_factorial}
  \label{matrix_factorial}
\end{algo}

\begin{proposition}\label{proof_cost_tree}
  This algorithm has a cost of
  $$\tilde{O}\big(m^{\omega}d N(n+d\log(N)+\log(m))\big)$$
  bit operations, where $n$ is the maximum bit size of the integers 
  in the matrix $M(\theta)$.
\end{proposition}
\begin{proof}
 The computation of the binary tree $S$ is less costly than that 
  of $T$, so
  we do not consider it. Let us evaluate the complexity of the computation of
  $T$. We need to know the bit size of the integers at each level of $T$. We
  use the following lemma.
  \begin{lemma}
    For any $a\leq N$, all the integers appearing in $M(\theta+a)$ have
    bit size at most $n+d(1+\log_2(N))$.
  \end{lemma}
  \begin{proof}
    Let $Q\in\mathbb{Z}[\theta]$ of degree less than $d$ appearing in 
    $M(\theta)$.
    Then we can write
    $$Q(\theta+a)=\sum_{j=0}^{d-1}\left(\sum_{i=j}^{d-
    1}\binom i j q_ia^{i-j}\right)\theta^j$$ 
      where the $q_i$ are the coefficients of $Q$. Moreover, we know that
      all the $q_i$ are at most $2^n$. Thus
      the coefficients of $Q(\theta+a)$ are less than~$2^n
      N^{d-1}\sum_{i=j}^{d-1}\binom i j\leq
        2^{n+d}N^d$.
  \end{proof}
  We now resume the proof of Proposition~\ref{proof_cost_tree}.
  If $\Delta_1$ and $\Delta_2$ are matrices in
  $\mathscr{M}_m\left(\nicefrac{\mathbb{Z}[\theta]}{\theta^d}\right)$ with
  integers of bit size at most $n_1$, then $\Delta_1\Delta_2$ has integers of bit size at most
  $2n_1+\log_2(dm)$. It follows that the integers in the matrices
  $A_{i,j}$ are of bit size at most:
  $$\begin{array}{l}
    2^{\eta-i}(n+d(1+\log_2(N)))+(2^{\eta-i}-1)\log_2(dm) \smallskip \\
    \hspace{5ex} = O(2^{\eta-i}(n+d\log_2(N)+\log_2(m))).
  \end{array}$$\par
  The computation of $T$ is reduced to the computation of its two sub-trees, 
  followed by a multiplication of two square matrices
  of size $m$ with polynomial coefficients of degree $d$ and integers of
  bit size~$O(2^{\eta-1}(n+d\log_2(N)+\log_2(m)))$.
  Since the bit size of the integers is halved at each
  level, we finally find, using that $2^{\eta}\leq 2N$, that the computation
  of $T$ can be done in $\tilde{O}(m^{\omega}d N(n+d\log_2(N)+\log_2(m))$
  bit operations.

  The cost of computing $W$
  is the same as that of reducing $T_{i,j}\mod S_{i,j+1}$ whenever both quantities
  are well defined, and then of computing recursively
  the $W_{i,j}$ using only integers smaller than~$S_{i,j}$. The first
  step can be done
  in $\tilde{O}(Nm^2d(n+d))$ bit operations, while the second requires
  $\tilde{O}(m^{\omega}d N)$ bit operations.
\end{proof}
\subsection{Final algorithm}\label{mainalgorithm}
The most important pieces of our main algorithm are now in place, we are almost
ready to write down its final version. Before doing this, we analyze the
cost of converting an operator in $\mathbb{Z}[x]\langle\partial\rangle$ to 
its counterpart in
$\mathbb{Z}[\theta]\langle\partial^{\pm 1}\rangle$.

\begin{proposition}
  For any operator $L\in\mathbb{Z}[x]\langle\partial\rangle$, of order $m$ 
  with coefficients of degree at most $d$, with integer coefficients of bit 
  size at most $n$, the computation of $\varphi(L)$, can be done in
  $\tilde{O}(d(m+d)(n+d))$ bit operations.\\
  Furthermore the resulting operator in the variable $\theta$ has its integer
  coefficients of bit size $O(n+d\log_2(d))$.
\end{proposition}
\begin{proof}
  From \cite[Section~4.1]{BoCaSc14} we get that this computation over a
  ring $R$ can be done
  in $\tilde{O}((m+d)d)$ algebraic operations in $R$. Following their
  algorithm, we can show that, when $R=\mathbb{Z}$, intermediate computations
  do not produce integers larger than those of the final result. Moreover, if
  $$\varphi\bigg(\sum_{\substack{0\leq i\leq d\\0\leq
  j\leq m}}l_{i,j}x^i\partial^j \bigg)=
  \sum_{\substack{0\leq i\leq d\\-d\leq j\leq
  m}}l'_{i,j}\theta^i\partial^j$$
  the estimation $|l_{i,j}|\leq 2^n$ implies $|l'_{i,j}|\leq 2^{n+d+1}d^d$.
  Putting all together, we get the announced result.
\end{proof}
  Note that for an operator $L\in\mathbb{Z}[x]\langle\partial\rangle$ of
  order $m$ with coefficients of degree at most $d$, $\varphi(L)$ has
  nonzero coefficients for powers of $\partial$ varying from $-d$ to $m$,
  making the square matrices used in Algorithm~\ref{finalalgorithm} of size
  at most $m+d$.\par
  We now present the final algorithm in Algorithm~\ref{finalalgorithm}.

\begin{algo}
  \begin{flushleft}
    \emph{Input:} $L_x\in\mathbb{Z}[x]\langle\partial\rangle$ of order $m$,
    with coefficients of degree at most~$d$ and integer coefficients of
    bit size at most $n$, $N\in\Nbb$.\\
    \emph{Output:} A list of polynomials $P_p\in\mathbb{F}_p[x,Y]$ such
    that $P_p(x^p,Y)=\chi(A_p(L))$ for all primes $p< N$, except a finite
    number not depending on $N$.
  \end{flushleft}
  \BlankLine
  \begin{enumerate}
    \item $l_x\leftarrow$ the leading coefficient of $L_x$.
    \item $a\leftarrow 0$.
    \item If $l_x(0)=0$ do:
      \begin{enumerate}
        \item Shift $L_x$ by $b$ with $b\in\mathbb{Z}$ not a root of
          $l_x$.
        \item $a\leftarrow b$.
      \end{enumerate}
      \textbf{Cost: }\emph{$\tilde{O}(md(n+d))$ bit operations.}
    \item Compute $L_\theta\partial^{-k}:=\varphi(L_x)$ with
      \texttt{x\_d\_to\_theta\_d} from \cite[Section~4]{BoCaSc14}.\\
      \textbf{Cost: }\emph{$\tilde{O}((m+d)(n+d)d)$ bit operations.}
    \item $d\leftarrow$ the maximum degree of the coefficients of
      $L_\theta$.
    \item $l_\theta\leftarrow$ the leading coefficient of $L_\theta$.\\
      \emph{It has been made to be an integer.}
    \item Construct $M(\theta)=l_\theta\cdot B(L_\theta)$.
    \item Compute the list $\mathcal{P}$ of all primes $p$ that do not
      divide $l_\theta$ with $d+1\leq p< N$.\\
      \textbf{Cost: }\emph{$\tilde{O}(N)$ bit operations (see 
      \cite[Proposition~2.1]{CoGeHa14}).}
    \item Compute the list $\mathcal{L}$ of
      $M(\theta)\cdots M(\theta+p-1)\bmod (\theta^{d+1}
      ,p)$ for all $p$ in $\mathcal{P}$ using \emph{matrix\_factorial}.\\
      \textbf{Cost: }\emph{$\tilde{O}((m+d)^\omega(n+d) d N)$ bit
      operations.}
    \item Divide all elements of $\mathcal{L}$ by $l_\theta$.\\
      \textbf{Cost: }\emph{$O(N(m+d)^2d$ bit operations.}
    \item Compute the list $\mathcal{C}$ of the characteristic polynomials of elements 
      of $\mathcal{L}$.\\
      \textbf{Cost:} \emph{$\tilde{O}(N(m+d)^{\Omega_1}d)$ bit operations.}
    \item Multiply the elements of $\mathcal{C}$ by $l_\theta$.\\
      \textbf{Cost: }\emph{$\tilde{O}(N(m+d)d)$ bit operations.}
    \item Compute the image by $\varphi_p^{-1}$ of elements of $\mathcal{C}$ using
      \emph{reverse\_iso}.\\
      \textbf{Cost: }\emph{$\tilde{O}(Nd(m+d))$ bit operations.}
    \item Divide the polynomials obtained by $l_x$ and $Y^{-k}$.\\
      \textbf{Cost: }\emph{$\tilde{O}(Nmd)$ bit operations.}
    \item If $a\neq 0$, shift the polynomials obtained by $-a$.\\
      \textbf{Cost: }\emph{$\tilde{O}(Nmd)$ bit operations.}
    \end{enumerate}
    \caption{charpoly\_p\_curv}
    \label{finalalgorithm}
\end{algo}
\begin{theorem}\label{finalcomplexity}
  For any operator $L\in\mathbb{Z}[x]\langle\partial\rangle$, 
  Algorithm~\ref{finalalgorithm}
  computes a list of polynomials $P_p\in\mathbb{Q}[x,Y]$ for all primes 
  $p< N$ except a finite number not depending on $N$, such that $P_p(x^p,
  Y)=\chi(A_p(L))$
  in 
  \[\tilde{O}\big(Nd((n+d)(m+d)^\omega+(m+d)^{\Omega_1})\big)\] bit
  operations, where $m$ is the order of the operator, $d$ is the maximum degree
  of its coefficients
  and $n$ is the maximum bit size of the integers appearing in $L$.
\end{theorem}
\begin{proof}
  This is easily seen by summing the cost of each step of
  Algorithm~\ref{finalalgorithm}.
  We observe that these complexities are correct
  whether or not~$0$ is a root of $L_x$. Indeed, when it is not, 
  the new operator obtained after the translation of step~(3) has
  integer coefficients of bit size $O(n+d\log(d))$, therefore our
  complexity analysis remains correct.
\end{proof}

As we have seen, Algorithm~\ref{finalalgorithm} does not compute the characteristic
polynomial of the $p$-curvature for every $p< N$, as we have to remove all primes dividing $l_x(0)$, where $l_x$ is the leading
coefficient of the operator (provided of course that $l_x(0)\neq 0$).
Primes less than the maximum degree of
the coefficients of the operator are also not included; however, it is
possible to remedy these with minor tweaks using Remark~\ref{remark_on_d}.
\begin{proposition}
  It is possible to compute all characteristic polynomials of the
  $p$-curvatures of an operator $L\in\mathbb{Z}[x]\langle\partial\rangle$
  of order~$m$ and maximum degree of the coefficients~$d$, for all primes~$p$
  less than~$N$, in asymptotically quasi-linear time in~$N$.
\end{proposition}
\begin{proof}
  The computation for primes dividing $l_x(0)$ (with $l_x$
  being the leading coefficient of $L$) can be done using the main
  algorithm from \cite{BoCaSc14}.
  All other primes can be addressed using our new
  Algorithm~\ref{finalalgorithm}.\\
  As primes which cannot be computed using our algorithm only depend on the
  operator itself, the result immediately follows.
\end{proof}

\section{Implementation and timings}
We have implemented Algorithm~\ref{finalalgorithm} in the Computer Algebra software \emph{SageMath}. The source code can be downloaded from the following URL:
\href{https://github.com/raphaelpagesub/p\_curvatures}{\tt https://github.com/raphaelpagesub/p\_curvatures}.

As mentioned earlier, the computation of the characteristic polynomial of a
matrix of size~$m$ with coefficients in a ring can be performed in theory
using $\tilde{O}(m^{\Omega_1})$ ring operations, with $\Omega_1\simeq
2.697263$, see~\cite{KaVi04}.
However, we did not
implement the algorithm from~\cite{KaVi04}, and instead used an algorithm
computing a Hessenberg form of the matrix in $O(m^3)$
operations \cite{CaRoVa17}. Indeed, the latter algorithm is easier to
implement and the
computation of the characteristic polynomials is usually not the bottleneck
and does not hinder the quasi-linear nature of our algorithm. Furthermore,
experiments, as well as Theorem~\ref{finalcomplexity}, showed that most of
the running time is spent on the
computation of trees $T$ and $W$ when the order of the operator is of the
same magnitude as the degrees of its coefficients. We expect this trend to
improve when the ratio of these two factors grows in favor of the
order of the operator, but all experiments conducted so far showed that the
computation of the characteristic polynomials is never the bottleneck by a wide
margin.
It is still more than six times faster on an operator of order~$50$ with
coefficients of degree~$2$, for~$N=100$.
\begin{remark}
  In our experiments we do not consider cases where the degree $d$ of the
  coefficients is higher
  than the order $m$ of the operator because the complexity in $d$ is worse
  than in $m$. As in \cite[Section~IV]{BeBoVdH12}, the general case reduces
  to this one using the transformation~$x\mapsto -\partial$,
  $\partial\mapsto x$ which exchanges the roles of $\partial$ and $x$.
\end{remark}
\subsection{Timings on random operators}
\paragraph{Quasilinear as expected.}
Figure~\ref{figure1} shows computation timings of our implementation for
operators in $\mathbb{Z}[x]\langle\partial\rangle$ of varying sizes on \emph{SageMath} version 9.3.rc4 on an Intel(R)
Core(TM) i3-40050 machine at 1.7Ghz, running ArchLinux. As expected,
it does appear that our algorithm finishes in quasi-linear time in $N$.
We can also see a floor phenomenon, with computation time varying very little
between two powers of $2$, and then doubling. This is an expected effect of
the use of the complete binary tree structure in our algorithm.
This effect however seems less visible, even if it is still perceptible, as
the operator size increases. This is probably due to the fact that for
operators of small sizes, the cost of manipulating empty nodes is
non-negligible.

\begin{figure}
  \includegraphics[width=0.9\linewidth]{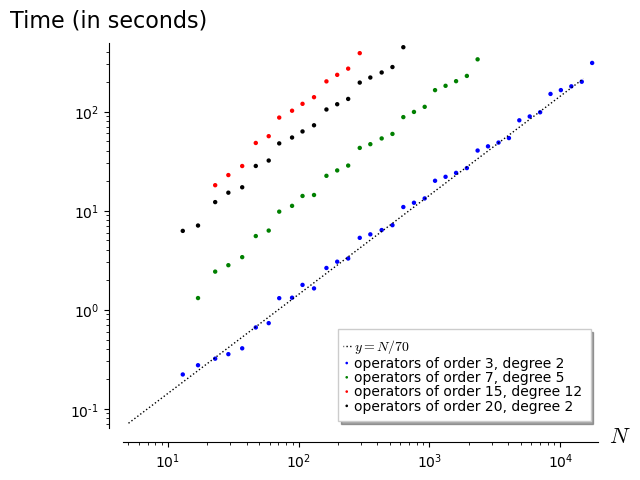}
  \caption{Computation time for random operators of varying orders and degrees}
  \label{figure1}
\end{figure}

\paragraph{Comparison with the previous algorithm.}

We have compared
the timings between our algorithm and the iteration of that of
\cite{BoCaSc14} for an operator of order 3 and degree 2.
Results are displayed on Figure~\ref{figure3} and show that the work
presented in this paper is indeed a
concrete progress for the considered task, compared to previous state of
the art:
experiments have shown that our algorithm was already more
than twice as fast (on the same machine) than the algorithm of \cite{BoCaSc14} \footnote{The implementation of the algorithm
from \cite{BoCaSc14} used can
be found at
\href{https://github.com/raphaelpagesub/p_curvatures/blob/main/p_curvature_single.sage}{\tt \scriptsize
https://github.com/raphaelpagesub/p\_curvatures/blob/main/p\_curvature\_single.sage}} for
$N\sim 10^4$. The right part
shows the ratio of computation times for operators of varying sizes.
Results tend to indicate that the good performances of our algorithm
compared to
the iteration of \cite{BoCaSc14} appear earlier when
the order of the operator grows. Further experiments should be
conducted to determine the influence of the degree of the coefficients.

\begin{figure*}
  \begin{tabular}{@{}c@{\hspace{4ex}}c@{}}
  \includegraphics[width=0.405\linewidth]{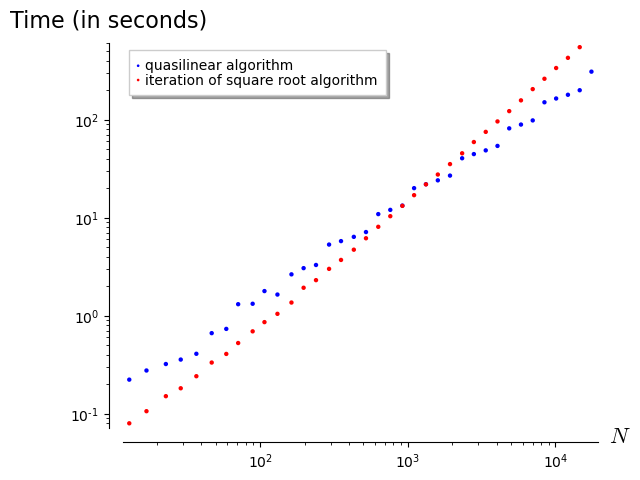} &
  \includegraphics[width=0.405\linewidth]{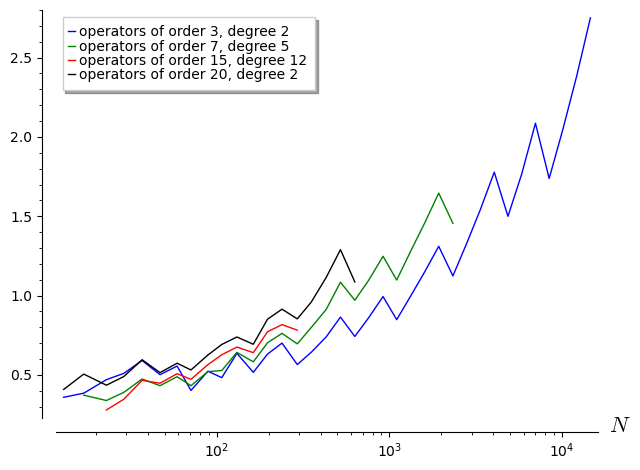} \\
  Computation time for operators of order 3 and degree 2 &
  Ratio of computation times for operators of varying sizes
  \end{tabular}
  \caption{Comparison between the iteration of
  \cite{BoCaSc14}'s algorithm and our algorithm}
  \label{figure3}
\end{figure*}

\subsection{Execution on special operators}
Our algorithm was also tested on various ``special'' operators. One example is
an
operator proven in~\cite{BoKa10} to annihilate the generating function 
$G(t;1,0)$ of Gessel walks in the quarter plane ending on the horizontal axis. 
The result of this test indicates that this operator has a nilpotent
$p$-curvature for all primes $p< 200$.  
This was of course expected since the generating function of Gessel walks is
algebraic~\cite{BoKa10}, hence the $p$-curvatures of its minimal-order
differential operator are all zero.
A similar test was performed on an operator proved in~\cite{BoKaVe20} 
to annihilate the generating function of Kreweras walks with interacting boundaries, which is not algebraic.   
Once again, the result of this test indicates that this operator has a nilpotent
$p$-curvature for all primes $p< 200$\footnote{The program running the above 
mentioned tests can be found at
\href{https://github.com/raphaelpagesub/p\_curvatures/blob/main/test\_p\_curvature.sage}
{\tt \scriptsize https://github.com/raphaelpagesub/p\_curvatures/blob/main/test\_p\_curvature.sage}}.
Further testing was conducted on all the 76 operators 
for (specializations of) the D-finite generating functions for lattice walks 
classified in \cite{BCHKP17} with $p< 200$, with yet again similar
results\footnote{The precise list of
operators we considered can be found at
\href{https://specfun.inria.fr/chyzak/ssw/ct-P.mpl}
{\tt \scriptsize https://specfun.inria.fr/chyzak/ssw/ct-P.mpl}
and the testing file can be found at
\href{https://github.com/raphaelpagesub/p\_curvatures/blob/main/ct-P.sage}
{\tt \scriptsize https://github.com/raphaelpagesub/p\_curvatures/blob/main/ct-P.sage}}.
All those results were already
predicted by Chudnovsky's theorem and make us quite confident in the
accuracy of our implementation.\par

\section{Conclusion and Future Work}\label{sec:conclusion}
We have proposed an algorithm which computes the characteristic polynomials
of the $p$-curvatures of a differential operator with coefficients in
$\mathbb{Z}[x]$ 
for almost all primes $p< N$,
in quasi-linear time in $N$.\par
We expect that the principle of this algorithm can theoretically be applied 
for differential operators with polynomial coefficients in any ring $A$ by replacing
$\nicefrac{\mathbb{Z}}{p\mathbb{Z}}$ by $\nicefrac{A}{pA}$.
Especially we expect that this algorithm extends nicely to
operators with polynomial coefficients in the integer ring of a number
field or with multivariate polynomial coefficients (which will allow us to
deal with
operators with parameters). In the latter case, we expect its
time complexity in $N$ to be in $\tilde{O}(N^s)$ where~$s$ is the
number of variables.
Furthermore, \cite{BoCaSc16} brought back the computation of the similarity
class of the $p$-curvature of an operator in $K[x]\langle\partial\rangle$,
with $K$ a field of positive characteristic, to that of a matrix factorial.
Thus we hope that the same principle can be applied to design an algorithm
for computing the similarity classes of the $p$-curvatures of an operator in
$\mathbb{Z}[x]\langle\partial\rangle$, for almost all primes $p< N$, in
quasi-linear time in $N$.\par
This algorithm may also have applications to future works on factorisation
of differential operators, as in \cite{Clu03}.

\bigskip

\def\gathen#1{{#1}}\def\cprime{$'$}
  \def\gathen#1{{#1}}\def\haesler#1{{#1}}\def\hoeij#1{{#1}}
\bibliographystyle{alpha}
\bibliography{./bibliography_issac}
\end{document}